\documentclass[11pt]{article}
\usepackage{amsthm,amsmath,amssymb}
\usepackage{subcaption}
\usepackage[usenames,dvipsnames]{color}
\usepackage[pdftex,breaklinks,colorlinks,
    citecolor={BlueViolet}, linkcolor={Blue},urlcolor=Maroon]{hyperref}
\usepackage{tikz,pgfkeys}
\usepackage{tkz-graph}
\usetikzlibrary{calc}
\usetikzlibrary{arrows.meta}
\usetikzlibrary{quotes}
\usetikzlibrary{decorations.pathreplacing}
\usepackage{graphicx}
\usepackage{charter,eulervm}%
\usepackage{multirow,booktabs,array}
\usepackage{tabularx,enumerate}
\usepackage[final,expansion=alltext,protrusion=true]{microtype}
\usepackage[colorinlistoftodos,prependcaption,textsize=tiny]{todonotes}
\usepackage{tikz}

\theoremstyle{plain} 
\newtheorem{theorem}{Theorem}[section]
\newtheorem{lemma}[theorem]{Lemma}
\newtheorem{corollary}[theorem]{Corollary}
\newtheorem{proposition}[theorem]{Proposition}

\tikzstyle{vertex}  = [{circle,blue,draw,fill=black!50,inner sep=1pt}]
\tikzstyle{class} = [shape=rectangle, minimum height=5mm, rounded corners, draw, align=center, top color=white, bottom color=blue!20]
\newcommand{\hsc}[1]{{\footnotesize\sf\MakeUppercase{#1}}}

\newcommand{\stpath}[2]{$#1$-$#2$ path}
\newcommand{\stsep}[2]{$#1$-$#2$ separator}
\newcommand{\lp}[1]{\ensuremath{{\mathtt{lp}(#1)}}}
\newcommand{\rp}[1]{\ensuremath{{\mathtt{rp}(#1)}}}

\DeclareMathOperator*{\argmin}{arg\,min}
\DeclareMathOperator*{\argmax}{arg\,max}

\title{Graph Searches and Their End Vertices}
\author{
    Yixin Cao\thanks{Department of Computing, Hong Kong Polytechnic University, Hong Kong, China. \href{mailto:yixin.cao@polyu.edu.hk} {\tt yixin.cao@polyu.edu.hk}.} 
  \and
  Guozhen Rong\thanks{School of Computer Science and Engineering, Central South University, Changsha, China.}
  \and
  Jianxin Wang\footnotemark[2]
  \and
Zhifeng Wang\footnotemark[2]
}
\date{May 23, 2019}

\begin{document}

\maketitle
\begin{abstract}
  Graph search, the process of visiting vertices in a graph in a specific order, has demonstrated magical powers in many important algorithms.  But a systematic study was only initiated by Corneil et al.~a decade ago, and only by then we started to realize how little we understand it.  Even the apparently na\"{i}ve question ``which vertex can be the last visited by a graph search algorithm,'' known as the end vertex problem, turns out to be quite elusive.  We give a full picture of all maximum cardinality searches on chordal graphs, which implies a polynomial-time algorithm for the end vertex problem of maximum cardinality search.  It is complemented by a proof of NP-completeness of the same problem on weakly chordal graphs.
  We also show linear-time algorithms for deciding end vertices of breadth-first searches on interval graphs, and end vertices of lexicographic depth-first searches on chordal graphs.  Finally, we present $2^n\cdot n^{O(1)}$-time algorithms for deciding the end vertices of breadth-first searches, depth-first searches, maximum cardinality searches, and maximum neighborhood searches on general graphs.
\end{abstract}

\section{Introduction}
Breadth-first search (\textsc{bfs}) and depth-first search (\textsc{dfs}) are the most fundamental graph algorithms, and the standard opening of a course on this subject.  Their use can be found, sometimes implicitly, in most graph algorithms.
In general, a graph search is a systematic exploration of a graph, and its core lies on the strategy of how to choose the next vertex to visit.  Mostly greedy, graph searches are very simple but sometimes have magical powers.  \textsc{Dfs} has played a significant role in Tarjan's award-winning work, in testing planarity \cite{hopcroft-74-planarity-testing} and in finding strongly connected components \cite{tarjan-72-dfs}.

Two other search algorithms, lexicographic breadth-first search (\textsc{lbfs}) \cite{rose-76-vertex-elimination} and maximum cardinality search (\textsc{mcs}) \cite{tarjan-84-chordal-recognition}, were invented for the purpose of recognizing chordal graphs, i.e., graphs not containing any induced cycle on four or more vertices.
On a chordal graph, both \textsc{lbfs} and \textsc{mcs} produce perfect elimination orderings (see definition in the next section) of the graph, which exist if and only if the graph is chordal.
Albeit relatively less well known compared to \textsc{bfs} and \textsc{dfs}, \textsc{lbfs} and \textsc{mcs} did find important applications.  \textsc{Lbfs} is used in scheduling \cite{sethi-76-scheduling-graphs}, and is the base of the recent linear-time algorithm for computing modular decomposition of a graph \cite{tedder-08}.  Tarjan and Yannakakis~\cite{tarjan-84-chordal-recognition} also used \textsc{mcs} in testing acyclic hypergraphs.  Nagamochi and Ibaraki~\cite{nagamochi-92-edge-connectivity} rediscovered \textsc{mcs} and applied it to compute minimum cuts of a graph and find forest decompositions; see also \cite{nagamochi-10-algorithmic-connectivity}.

Simon~\cite{simon-91-interval} proposed an interesting way of using \textsc{lbfs}.  It conducts \textsc{lbfs} more than once, and each new run uses previous runs in breaking ties; in particular, except the first, each run starts from the last vertex of the previous run.  This generic approach turns out to be very useful, e.g., the extremely simple recognition algorithm for unit interval graphs \cite{corneil-04-recognize-uig}.  See the survey of Corneil~\cite{corneil-04-survey-lbfs} for more algorithms using multiple runs of \textsc{lbfs}.  Some of these results have a flavor of ``ad-hoc'': We do not fully understand the execution process of \textsc{lbfs}.

The outputs of \textsc{bfs} and \textsc{dfs} are usually rooted spanning forests of the graph, while \textsc{lbfs} and \textsc{mcs} produce orderings of its vertices.  To have a unified view of them, Corneil et al.~\cite{corneil-08-graph-searching} focused on the ordering of the vertices being first visited and conducted a systematic study of them.\footnote{One may note that we can define more than one vertex ordering for \textsc{dfs} and accordingly \textsc{ldfs}.  See \cite{corneil-08-graph-searching}.  On the other hand, the orderings produced by \textsc{lbfs} and \textsc{mcs} are conventionally the reverse of ours: Recall that their original purpose is to produce perfect elimination orderings of a chordal graph, which are reverse of our orderings.} This study motivates them to propose the lexicographic version of \textsc{dfs}, lexicographic depth-first search (\textsc{ldfs}), another very powerful graph search \cite{corneil-13-path-cover-on-cocomparability}, and a very general search paradigm, maximum neighborhood search (\textsc{mns}).  They showed that all the aforementioned graph searches can be characterized by variants of the so-called four-vertex condition.
These nice characterizations are however not sufficient to allow us to answer the ostentatiously na\"{i}ve question: Which vertex can be the last of such an ordering?  Corneil et al.~\cite{corneil-10-end-vertices-lbfs} defined the end vertex problem and studied it from both combinatorial and algorithmic perspectives.  Apart from a natural starting point of understanding the graph searches in general, end vertices of graph searches are of their own interest.  Behind the original use of \textsc{lbfs} and \textsc{mcs}, in the recognition of chordal graphs, is nothing but the properties of their end vertices, which are always simplicial on a chordal graph \cite{rose-76-vertex-elimination, tarjan-84-chordal-recognition, shier-84-all-perfect-elimination-orderings, berry-98-generalizes-dirac}.  Moreover, the success of multiple-run \textsc{lbfs} crucially hinges on the end vertices; e.g., an end vertex of a (unit) interval graph can always be assigned an extreme (i.e., leftmost or rightmost) interval \cite{corneil-04-recognize-uig, corneil-09-lbfs-strucuture-and-interval-recognition}.  Important properties and use of end vertices of other graph searches can be found in \cite{corneil-13-path-cover-on-cocomparability, habib-00-LBFS-and-partition-refinement, corneil-99-dominating-pairs-in-at-free}.

One may find it surprising, but the end vertex problem is NP-hard for all the six mentioned graph search algorithms \cite{corneil-08-graph-searching, charbit-14-tie-break-rule, beisegel-18-end-vertex}.  The study has thus been focused on chordal graphs and its closely related superclasses and subclasses.
After all, \textsc{lbfs} and \textsc{mcs} were invented for recognition of chordal graphs, and their properties on chordal graphs have been intensively studied.  (This renders the stagnation on chordal graphs a little more embarrassing.)  Moreover, most applications of \textsc{lbfs} and \textsc{ldfs} are on related graph classes.
The most natural superclass of chordal graphs is arguably the weakly chordal graphs, and two important subclasses are interval graphs and split graphs.  
It has been known that on weakly chordal graphs, the end vertex problems for all but \textsc{mcs} are NP-complete, while only \textsc{dfs} end vertex is NP-complete on chordal graphs \cite{corneil-08-graph-searching, charbit-14-tie-break-rule, beisegel-18-end-vertex}.  There are other polynomial-time algorithms for interval graphs and split graphs, most of which actually run in linear time. We complete the pictures for, in terms of graph searches, \textsc{mcs} and \textsc{ldfs}, and, in terms of graph classes, weakly chordal graphs and interval graphs.
A summary of known results is given in Figure~\ref{fig:overview}.

\begin{figure}[h]
  \centering\small
  \begin{tikzpicture}[every path/.style={thick}, scale = 1.2]
    \node[class] (weak) at (2, 8) {\hsc{weakly chordal}};
    \node[class] (chordal) at (2, 6.5) {\hsc{chordal}};
    \node[class] (split) at (4.2, 5) {\hsc{split}};
    \node[class] (interval) at (-.2, 5) {\hsc{interval}};
    \draw[violet] (split) -- (chordal);
    \draw[violet] (interval) -- (chordal) -- (weak); 
    \node[left = 0.25cm of interval] {all~\cite{corneil-10-end-vertices-lbfs, beisegel-18-end-vertex}}; 
    \node[left = 0.25cm of chordal] {\textsc{mns~\cite{beisegel-18-end-vertex}, mcs, ldfs}};
    \node[right = 0.25cm of chordal] {\textsc{dfs}};
    \node[left = 0.1cm of split] {all others~\cite{charbit-14-tie-break-rule, beisegel-18-end-vertex}};
    \node[right = 0.15cm of split] {\textsc{dfs}~\cite{charbit-14-tie-break-rule}};
    \node[right = 0.15cm of weak] {all~\cite{corneil-10-end-vertices-lbfs, charbit-14-tie-break-rule, beisegel-18-end-vertex}};

    \node (p) at (-3, 8) {P};
    \coordinate[right = 0.2cm of p] (x);
    \draw[-latex,violet] (x) --++ (0, -5mm);
    \node (np) at (7, 5) {NPC};
    \coordinate[left = 0.2cm of np] (y);
    \draw[-latex,violet] (y) --++ (0, 5mm);
  \end{tikzpicture}
  \caption{A summary of the known complexity of the end vertex problem for the six graph search algorithms.  For each graph class, the end vertex problem of graph searches listed to the left of it can be solved in polynomial time on this class, while those to the right are NP-hard.  The complexity of the  \textsc{bfs} end vertex and \textsc{lbfs} end vertex problems on chordal graphs are still open.}
  \label{fig:overview}
\end{figure}
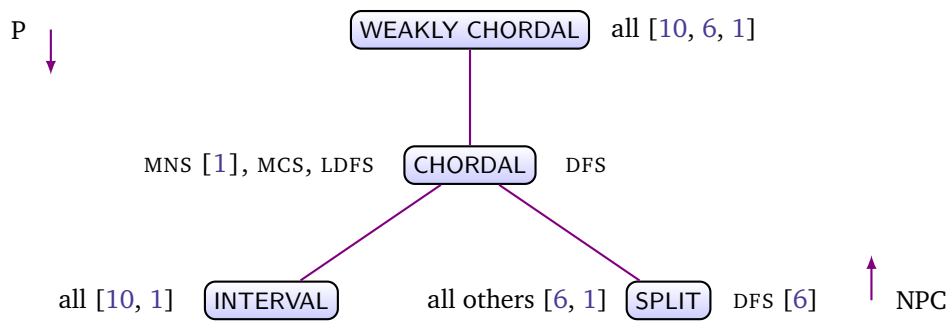

Blair and Peyton~\cite{blair-91-chordal-graphs-clique-trees} and Galinier et al.~\cite{galinier-95-chordal-graphs-clique-graphs} have shown that \textsc{mcs} of a chordal graph are closely related to its maximal cliques.  Let $G$ be a chordal graph.  An \textsc{mcs} visits all vertices in a maximal clique of $G$ before proceeding to another, and the next maximal clique is always chosen to have the largest intersection with a visited one.
Therefore, for a minimum separator $S$ of $G$, there is an \textsc{mcs} visiting the components of $G - S$ one by one, with $S$ visited together with the first component.  If we turn to any component $C$ of $G - S$, and consider its closed neighborhood, (which contains $C$ and $S$,) then we have a similar statement.  In other words, this property on minimum separators hold in a recursive way.  For an \textsc{mcs} end vertex $z$, which is necessarily simplicial, we can find a sequence of increasing separators such that the first is a minimum separator of $G$ and the last comprises all the non-simplicial vertices in $N(z)$.  An \textsc{mcs} ended with $z$ has to ``cross'' these separators in order, and for each of them, visit the component containing $z$ in the last.  We have thus a full understanding of all \textsc{mcs} orderings of a chordal graph.  As it turns out, this result is easier to be presented in the so-called weighted clique graph of $G$ \cite{blair-91-chordal-graphs-clique-trees, galinier-95-chordal-graphs-clique-graphs}.  It enables us to show that if we run \textsc{mcs} twice, first starting from $z$, and the second starting from the end vertex of the first run and using the first ordering to break ties, then the second run ends with $z$ if and only if $z$ is an \textsc{mcs} end vertex.  As usual, $n$ denotes the number of vertices in the input graph.

\begin{theorem}\label{thm:mcs-chordal}
  The \textsc{mcs} end vertex problem can be solved in $O(n^2)$ time on chordal graphs.
\end{theorem}

We complement this result by showing that the \textsc{mcs} end vertex problem becomes NP-complete on weakly chordal graphs;  the proof is inspired by and adapted from Beisegel et al.~\cite{beisegel-18-end-vertex}.

\begin{theorem}\label{thm:mcs-weakly-chordal}
  The \textsc{mcs} end vertex problem is NP-complete on weakly chordal graphs.
\end{theorem}

We then turn to \textsc{ldfs} on chordal graphs.  Surprisingly, the characterization of Berry et al.~\cite{berry-10-extremities-search} for end vertices of \textsc{mns} on chordal graphs is also true for \textsc{ldfs}:  A simplicial vertex $z$ of a chordal graph $G$ is an \textsc{ldfs} end vertex if and only if the minimal separators of $G$ in $N(z)$ are totally ordered by inclusion.  We also show a simple algorithm for solving the \textsc{bfs} end vertex problem on interval graphs.  

\begin{theorem}
  There are linear-time algorithms for solving the \textsc{ldfs} end vertex problem on chordal graphs and for solving the \textsc{bfs} end vertex problem on interval graphs.
\end{theorem}

We have to, nevertheless, leave open the \textsc{bfs} and \textsc{lbfs} end vertex problems on chordal graphs.  Since both can be solved in linear time on split graphs, we conjecture that they can be solved in polynomial time on chordal graphs.  It is extremely rare that a problem is hard on chordal graphs but easy on split graphs.

We also consider algorithms for solving the end vertex problems on general graphs.  By enumerating all possible orderings, a trivial algorithm can find all end vertices of any graph search in $n!\cdot n^{O(1)}$ time.  On the other hand, with the only exception of \textsc{bfs}, the reductions used in proving NP-hardness of the end vertex problems are linear reductions from (3-)\textsc{sat}.
As a result, these problems cannot be solved in subexponential time, unless the exponential time hypothesis fails \cite{impagliazzo-01-eth}.  A natural question is thus which of them can be solved in $2^{O(n)}$ time.  If we put them under closer scrutiny, we will see that these graph searches are somewhat different: When selecting the next vertex, \textsc{mcs} only needs to know which vertices have been visited, while the order of visiting them is immaterial.  In contrast, the other graph searches are not \emph{oblivious} and need to keep track of the whole visiting history.  Therefore, it is quite straightforward to use dynamic programming to solve the \textsc{mcs} end vertex problem in $2^{n}\cdot n^2$ time.  We also manage to show that a similar approach actually works for  the \textsc{bfs} and \textsc{dfs} end vertex problems. 

\begin{theorem}\label{thm:exact-alg}
  There are $2^n\cdot n^{O(1)}$-time algorithms that solve the end vertex problems of the following graph searches: \textsc{mcs}, \textsc{bfs}, and \textsc{dfs}.
\end{theorem}

\section{Preliminaries}

All graphs discussed in this paper are undirected and simple.  The vertex set and edge set of a graph $G$ are denoted by, respectively, $V(G)$ and $E(G)$, and we use $n = |V(G)|$ and $m = |E(G)|$ to denote their {cardinalities}.
For a subset $X\subseteq V(G)$, denote by $G[X]$ the subgraph of $G$ induced by $X$, and by $G - X$ the subgraph $G[V(G)\setminus X]$. 
The \emph{(open) neighborhood} of a vertex $v \in V(G)$, denoted by $N(v)$, comprises vertices adjacent to $v$, i.e., $N(v) = \{ u \mid uv \in E(G) \}$, and the \emph{closed neighborhood} of $v$ is $N[v] = N(v) \cup \{ v \}$.  The \emph{degree} of a vertex $v$ is the number of neighbors it has, i.e., $d(v)=|N(v)|$.
A vertex $v$ is \emph{simplicial} if $N[v]$ induces a complete graph.
Two distinct vertices $u$ and $v$ are \emph{true twins} if $N[u] = N[v]$, and \emph{false twins} if $N(u) = N(v)$; note that true twins are adjacent while false twins are not.

A set $S$ of vertices is a \emph{\stsep{u}{v}} if $u$ and $v$ are not in $S$ and they are not connected in $G-S$, and a \emph{\stsep{u}{v}} is \emph{minimal} if no proper subset of $S$ is a \emph{\stsep{u}{v}}.  We say that $S$ is a (minimal) separator if it is a (minimal) \stsep{u}{v} for some pair of $u$ and $v$, and it is a \emph{minimum separator} of $G$ if it has the smallest cardinality among all separators of $G$.

An \emph{ordering} $\sigma$ of the vertices of $G$ is a bijection from $V(G)\to \{1, \ldots, n\}$.  For two vertices $u$ and $v$, we use $u<_\sigma v$ to denote $\sigma(u) < \sigma(v)$. 
The \emph{end vertex} of $\sigma$ is the vertex $z$ with $\sigma(z) = n$.  Given a graph $G$ and a vertex $z\in V(G)$, the \emph{end vertex problem} for graph search $S$ is to determine whether there is an $S$-ordering of $G$ of which $z$ is the end vertex.

A graph is \emph{chordal} if it contains no induced cycle on four or more vertices.  A graph is chordal if and only if it can be made empty by removing simplicial vertices from the remaining graph one by one; the order of the vertices removed is called a \emph{perfect elimination ordering} \cite{fulkerson-65-interval-graphs}.
The greedy strategy of \textsc{mcs} is to choose an unvisited vertex with the maximum number of visited neighbors.
On a chordal graph $G$, the last vertex of any \textsc{mcs} is simplicial, and thus the reversal of an \textsc{mcs} ordering is always a perfect elimination ordering \cite{tarjan-84-chordal-recognition}.

To avoid unnecessary digressions, we consider only connected graphs.  All the results can be easily generalized to general graphs.

\section{Maximum cardinality search on chordal graphs}
Another important characterization of chordal graphs is through its maximal cliques.  A graph $G$ is chordal if and only if we can arrange its maximal cliques as a tree such that for each vertex $v\in V(G)$, maximal cliques containing $v$ induce a subtree; such a  tree is called a \emph{clique tree} of $G$ \cite{dirac-61-chordal-graphs}.  A chordal graph $G$ has at most $n$ maximal cliques \cite{dirac-61-chordal-graphs}, and for any pair of adjacent $K_i$ and $K_j$ on the clique tree, the intersection $K_i \cap K_j$ is a minimal separator of $G$.

Out of a chordal graph $G$, we can define a \emph{weighted clique graph} $C(G)$ as follows.  It has $\ell$ vertices, where $\ell$ is the number of maximal cliques of $G$, and each vertex is labeled by a distinct maximal clique of $G$.
 To simplify the presentation, we will refer to vertices of $C(G)$ as cliques; note that we are not going to use cliques of the graph $C(G)$ in this paper.
There is an edge between maximal cliques $K_i$ and $K_j$, $1\le i, j \le \ell$, if and only if $K_i \cap K_j$ is a minimal \stsep{x}{y} for all $x\in K_i \setminus K_j$ and $y\in K_j \setminus K_i$.   We label this edge with $K_i \cap K_j$, and set its weight to be $|K_i \cap K_j|$.
It is known that a tree on the maximal cliques of $G$ is a clique tree of $G$ if and only if it is a maximum spanning tree of $C(G)$ \cite{bernstein-81-natural-semijoins, blair-91-chordal-graphs-clique-trees, galinier-95-chordal-graphs-clique-graphs}, i.e., a spanning tree of $C(G)$ with the maximum total edge weights.

\begin{proposition}\label{lem:minimal-separators}
  Let $G$ be a chordal graph and $C(G)$ the weighted clique graph of $G$.  A vertex set $S\subseteq V(G)$ is a minimal separator of $G$ if and only if it is the label for some edge of $C(G)$.
\end{proposition}
\begin{proof}
  The if direction is from the definition of $C(G)$.  Now suppose that $S$ is a minimal separator of $G$.  According to Blair and Peyton \cite[Theorem~4.3]{blair-91-chordal-graphs-clique-trees}, any clique tree of $G$ has two adjacent cliques whose intersection is $S$.  Since this clique tree is a subgraph of $C(G)$, there is an edge of $C(G)$ with label $S$.
\end{proof}

One can use Prim's algorithm to find a maximum spanning tree of $G$.  (Although proposed for the purpose of finding a minimum spanning tree, Prim's algorithm can be easily modified to find a maximum one.)  Starting from an arbitrary clique, it grows the tree by including one edge and one clique at a time, while the edge is chosen to have the largest weight among those crossing the partial tree that has been built, i.e., with one end in the current tree and the other not.  In the same spirit of graph search orderings, we can define a \emph{Prim ordering} to be the order maximal cliques of $G$ being included (visited) by Prim's algorithm, applied to $C(G)$.

Let $\pi$ be an ordering of the maximal cliques of $G$.  We say that an ordering $\sigma$ of $V(G)$ is \emph{generated by $\pi$} if $K_u <_\pi K_v$ implies $u <_\sigma v$, where $K_u$ and $K_v$ are the first maximal cliques in $\pi$ containing $u$, and respectively, $v$.  If $\pi = \langle K_1, K_2, \ldots, K_\ell\rangle$ and $c_i = |K_i\setminus \bigcup^{i-1}_{j=1}K_j|$ for $1\le i \le \ell$, then $\sigma$ can be represented as
\[
  \underbrace{\sigma^{-1}(1), \ldots, \sigma^{-1}(c_1)}_{K_1}, \; \underbrace{\sigma^{-1}(c_1 + 1), \ldots, \sigma^{-1}(c_1+c_2)}_{K_2\setminus K_1},\; \ldots, \underbrace{\sigma^{-1}(n-c_\ell+1), \ldots, \sigma^{-1}(n)}_{K_\ell\setminus \bigcup^{\ell-1}_{j=1}K_j}.
\]
The following has been essentially observed by Blair and Peyton~\cite{blair-91-chordal-graphs-clique-trees}, who however only stated explicitly one direction.  For the sake of completeness, we give a proof here.

\begin{lemma}
  \label{lemma:spanning-tree}
  Let $G$ be a chordal graph.  An ordering $\sigma$ of $V(G)$ is an \textsc{mcs} ordering of $G$ if and only if it is generated by some Prim ordering $\pi$ of $C(G)$.
\end{lemma}
\begin{proof}
  The only if direction has been proved by Blair and Peyton~\cite[Lemma 4.8 and Theorem 4.10]{blair-91-chordal-graphs-clique-trees}.  Here we show the if direction.
Suppose that $\sigma$ is generated by $\pi$.  We may renumber the vertices in $G$ such that $\sigma = \langle v_1$, $v_2$, $\ldots$, $v_n\rangle$, and renumber the maximal cliques such that $\pi = \langle K_1$, $K_2$, $\ldots$, $K_\ell\rangle$.  Let $K'_i = K_i\setminus \bigcup^{i-1}_{j=1}K_j$ for $1\le i\le \ell$; note that $\{K'_1$, $K'_2$, $\ldots$, $K'_\ell\}$ is a partition of $V(G)$.
We show by induction that for each $1\le i\le n$, there is an \textsc{mcs} ordering of $G$ of which the first $i$ vertices are $v_1, \ldots, v_i$; in other words, among vertices $v_{i}$, $\ldots$, $v_n$, vertex $v_i$ has the maximum number of neighbors in the first $i - 1$ vertices.  It is vacuously true for $i = 1$.  Now suppose that it is true for $v_p$, we show that it is also true for $v_{p + 1}$.

When $v_{p + 1}\in K'_1 = K_1$, it is adjacent to all previous vertices and we are done.  In the rest $v_{p + 1}\in K'_t$ for some $t > 1$.  Let $A = \bigcup^{t - 1}_{j=1}K_j$; note that $v_{p + 1}\not\in A$.
For any $q> p$, let $G_q$ denote the the subgraph of $G$ induced by $v_1, v_2, \ldots, v_{p}$, and $v_q$.  By the induction hypothesis, $\langle v_1, v_2, \ldots, v_{p}, v_q\rangle$ is an \textsc{mcs} ordering of $G_q$.  Since $G_q$ is chordal, $v_q$ is simplicial in it.  Therefore, $N(v_q)\cap A$ is a clique for all $q > p$; denote it by $X_q$.  We argue by contradiction that there must be $1\le s < t$ such that $X_q \subseteq K_s$.  We find an $i$ with $1\le i < t$ such that $K_i\cap X_q$ is maximal.  If $X_q \not\subseteq K_i$, then there is a vertex $x\in X_q\setminus K_i$; let $K_j$, where $1\le j < t$, contain $x$.  By the maximality of $K_i\cap X_q$, there exists $y\in (X_q\cap K_i)\setminus K_j$.  Of the first $t - 1$ maximal cliques, those containing $K_i\cap X_q$ and those containing $X_q\setminus K_i$ are disjoint.  Prim's algorithm always maintains a tree of visited cliques, and this tree is a subtree of a clique tree of $G$.  Therefore, there is an \stsep{x}{y}.  But this is impossible because $x$ and $y$ are both in $X_q$, hence adjacent.

For each $q > p$, there is some maximal clique $K$ of $G$ that contains $(N(v_q)\cap A)\cup \{v_q\}$.  It cannot be one of $K_1$, $\ldots$, $K_{t - 1}$ because $v_q\not\in A$.  Since $K_1, \ldots, K_\ell$ is a Prim ordering of $C(G)$, we have $|N(v_{p + 1})\cap A| \ge |N(v_{q})\cap A|$ for all $q > p$.  On the other hand, $v_{p + 1}$ is adjacent to all vertices in $K_t$.   We can thus conclude that $v_{p + 1}$ has the maximum number of neighbors in $\{v_1, \ldots, v_p\}$, and this completes the proof.
\end{proof}

By Lemma~\ref{lemma:spanning-tree}, \textsc{mcs} orderings of a chordal graph $G$ can be fully characterized by Prim orderings of its weighted clique graph $C(G)$.  In particular, the \textsc{mcs} end vertices are the private vertices of the cliques last visited by Prim's algorithm.
Note that a vertex $v$ is simplicial if and only if it belongs to precisely one maximal clique, namely, $N[v]$, and a set of true twins can be visited in any order.

\begin{corollary}\label{cor:weighted-clique-tree}
Let  $z$ be a simplicial vertex in a chordal graph $G$.  There exists an \textsc{mcs} ordering of $G$ ended with $z$ if and only if there exists a Prim ordering of $C(G)$ ended with $N[z]$.
\end{corollary}

Let $S$ be a separator of $G$.  We abuse notation to use $C(G) - S$ to denote the subgraph of $C(G)$ obtained by deleting all edges whose labels are subsets of $S$.  The component of $C(G) - S$ containing $N[z]$ is called the $z$-component of $C(G) - S$.  It is worth noting that $C(G) - S$ cannot be mapped back to $G$.  In Figure~\ref{fig:chordal-example}, for example, $C(G) - \{v_5, v_6\}$ does not have edges among $K_2$, $\ldots$, $K_5$, while edges $K_7 K_8, K_7 K_9, K_8 K_9, K_9 K_{10}$ will be removed in $C(G) - \{v_{12}, v_{13}\}$.

\tikzstyle{clique}  = [{circle,blue,draw, inner sep = 0, minimum size=6mm}]
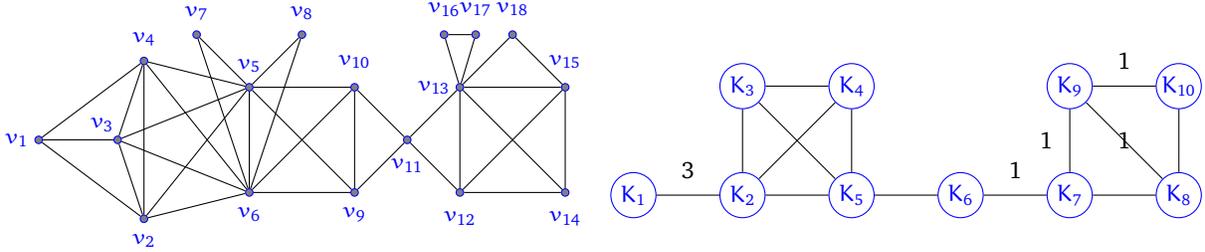
\begin{figure}[h]
  \centering
  \begin{tikzpicture}[every node/.style={vertex}, scale = .7]
    \scriptsize
    \foreach[count = \i from 1] \v/\l in {(0, 1)/left, (2, -.5)/below, (1.5, 1)/above left, (2, 2.5)/, (4, 2)/, (4, 0)/below, (3., 3)/, (5, 3)/, (6, 0)/below, (6, 2)/, (7, 1)/below, (8, 0)/below, (8, 2)/left, (10, 0)/below, (10, 2)/, (7.7, 3)/, (8.3, 3)/, (9, 3)/}
    \node["$v_{\i}$" \l] (v\i) at \v {};
    \foreach \i in {1, 5, 6}
    \foreach \j in {2, 3, 4}
    \draw (v\i) -- (v\j);
    \foreach \i in {5, 6}
    \foreach \j in {7, ..., 10}
    \draw (v\i) -- (v\j);
    \draw (v2) -- (v3) -- (v4) -- (v2);
    \draw (v5) -- (v6);
    \draw (v9) -- (v10) -- (v11) -- (v9);
    \foreach \i in {12, 13}
    \foreach \j in {11, 14, 15}
    \draw (v\i) -- (v\j);
    \draw (v13) -- (v16) -- (v17) -- (v13);
    \draw (v12) -- (v13) -- (v18) -- (v15) -- (v14);
  \end{tikzpicture}
  \;
  \begin{tikzpicture}[scale=1.45]
    \footnotesize
    \foreach[count = \i from 1] \v/\l in {(1, 1)/left, (2, 1)/below, (2, 2)/, (3, 2)/, (3, 1)/, (4, 1)/below, (5, 1)/, (6, 1)/below, (5, 2)/, (6, 2)/below}
    \node[clique] (k\i) at \v {$K_{\i}$};
    \draw (k1) --node["$3$"] {} (k2) -- (k3) -- (k4) -- (k5) -- (k2) -- (k4);
    \draw (k3) -- (k5) -- (k6) --node["$1$"] {} (k7) -- (k8) --node {$1$} (k9) --node["$1$"] {} (k10) -- (k8);
    \draw (k7) --node["$1$" left] {} (k9);
    \node at (2, 0.5) {};    
  \end{tikzpicture}
  \caption{A chordal graph $G$ on $18$ vertices (the left), and its weighted clique graph (the right), where all the omitted edge weights are $2$.  There are 10 maximal cliques $K_1 = \{v_1, v_2, v_3, v_4\}, K_2 = \{v_2, \ldots, v_6\}, K_3 = \{v_5, v_6, v_7\}, K_4 = \{v_5, v_6, v_8\}, K_5 = \{v_5, v_6, v_9, v_{10}\}, K_6 = \{v_9, v_{10}, v_{11}\}, K_7 = \{v_{11}, v_{12}, v_{13}\}, K_8 = \{v_{12}, \ldots, v_{15}\}, K_9 = \{v_{13}, v_{16}, v_{17}\}, K_{10} = \{v_{13}, v_{15}, v_{18}\}$.   There are $7$ simplicial vertices $v_1$, $v_7$, $v_8$, $v_{14}$, $v_{16}$, $v_{17}$, $v_{18}$, of which $v_{14}$ and $v_{18}$ are not \textsc{mcs} end vertices.}
  \label{fig:chordal-example}
\end{figure}

\begin{proposition}\label{lem:separators-on-clique-graph}
  Let $S$ be a separator of a chordal graph $G$.  For any vertex $v\not\in S$, maximal cliques containing $v$ remain connected in $C(G) - S$.   For any two distinct vertices $u, v\not\in S$, maximal cliques containing $u$ and $v$ are not connected in $C(G) - S$ if and only if $S$ is a \stsep{u}{v}.
\end{proposition}
\begin{proof}
  By definition, the maximal cliques containing $v$ are connected in any clique tree of $G$.  Since a clique tree of $G$ is a subgraph of $C(G)$, these cliques also induce a connected subgraph in $C(G)$.  For any edge in this subgraph, its label contains $v$, hence not a subset of $S$.  Therefore, these cliques induce the same connected subgraph in $C(G) - S$ as in $C(G)$.

  For the second assertion, we may assume  $uv\not\in E(G)$: Both sides are trivially false when $uv\in E(G)$.
Suppose to the contradiction of the if direction that there is a path $K_0,\ldots,K_p$ in $C(G) - S$ such that $u\in K_0$ and $v\in K_p$ while $u, v\not\in K_i$ for $0 < i < p$.   For each $1\le i \le p$, we can find a vertex $x_i\in (K_{i-1}\cap K_i)\setminus S$.  (These $p$ vertices may or may not be distinct.)  Then $u x_1, x_p v\in E(G)$, while $x_i$ and $x_{i + 1}$ are either the same or adjacent for all $1\le i < p$.  We have thus a \stpath{u}{v} in $G$ avoiding $S$, contradiction that $S$ is a \stsep{u}{v}.

We now consider the only if direction.  Let $u=x_0, x_1, \ldots, x_p= v$ be any \stpath{u}{v} in $G$.  Note that for each $0\le i\le p$, maximal cliques containing $x_i$ induce a connected subgraph, while for each $1\le j\le p$, there is a maximal clique containing both $x_{j - 1}$ and $x_j$.  We can find a path in $C(G)$ of which one end contains $u$ and the other contains $v$.  For each edge on this path, its label contains one of $x_i$, $0< i < p$.  Since maximal cliques containing $u$ and $v$ are not connected in $C(G) - S$, the label of at least one edge on this path is a subset of $S$.  By the first assertion, at least one of $x_1, \ldots x_{p-1}$ is in $S$.  In other words, every \stpath{u}{v} intersects $S$.  Therefore, $S$ is a \stsep{u}{v}.  This concludes the proof.
\end{proof}

We say that a minimum-weight edge $e$ of $C(G)$---by Proposition~\ref{lem:minimal-separators}, its label is a minimum separator of $G$,---is a \emph{critical edge} for maximal clique $K$ if one end of $e$ is in the same component as $K$ after all minimum-weight edges, including $e$, are removed from $C(G)$.  In other words, there is a path connecting $K$ and $e$ on which every edge has weight larger than $e$.  
In Figure~\ref{fig:chordal-example}, for example, $K_6 K_7$ is a critical edge for all cliques but $K_9$, while $K_8 K_9$ and $K_{10} K_9$ are critical edges for $K_8$ and $K_{10}$ respectively.  The following fact explains ``critical'' in the name.

\begin{proposition}\label{cor:critical-edge}
  Let $z$ be a simplicial vertex of a connected chordal graph $G$, and let $S_1$, $\ldots$, $S_k$ be the labels of all critical edges for $N[z]$.  In any Prim ordering of $C(G)$, cliques in the $z$-component of $C(G) - S_1 - \cdots - S_k$ appear consecutively.
    Moreover, if $S_1 = \cdots = S_k$, then the $z$-component of $C(G) - S_1$ can be visited in the end.
\end{proposition}
\begin{proof}
  Note that $C(G)$ is connected since $G$ is connected.  Let $T$ denote the $z$-component of $C(G) - S_1 - \cdots -S_k$.  Being minimum separators of $G$, all of $S_1$, $\ldots$, $S_k$ have the same size; let it be $t$.
  Note that the weight of every edge in $T$ is strictly larger than $t$; otherwise, we can find a path from $N[z]$ to such an edge in $T$, and identify another critical edge for $N[z]$ on this path. 

Let $\pi$ be any Prim ordering of $C(G)$.  We consider the first maximal clique $K$ in $T$ visited by $\pi$.  If $\pi(K) \ne 1$, the edge leading to $K$ has weight $t$.  By Prim's algorithm, when $K$ is visited, for each clique $K'$ with $K' <_\pi K$, all the edges between $K'$ and its unvisited neighbors have weight $t$.  All edges between $T$ and other components have weight $t$ as well, while all edges inside $T$ have weight $> t$.  Therefore, the maximal cliques in $T$ must be finished before a clique out of $T$ is visited.  This concludes the first assertion.

For the second assertion, suppose that $S = S_1 = \cdots =S_k$.  We give a Prim ordering that visits cliques in $T$ in the end.  It starts from a clique not in $T$, and it suffices to show that all cliques out of $T$ have been visited before the first in $T$.  By the definition of $C(G)$, in each component of $C(G) - S$, there is a maximal clique containing $S$.  Therefore, by {Proposition}~\ref{lem:separators-on-clique-graph}, there is an edge with label $S$ between any two components of $C(G) - S$.  In other words, the cliques not in $T$ are connected in $C(G)$.  Since the edges connecting $T$ and other components of $C(G) - S$ have weight $t$, the minimum in $C(G)$,  Prim's algorithm can always choose another edge.  Therefore, we can finish them before entering $T$.
\end{proof}

Whether a simplicial vertex $z$ can be an \textsc{mcs} end vertex turns out to be closely related to the critical edges for $N[z]$.  We first present a necessary condition, which is not satisfied by $v_{14}$ and $v_{18}$ in Figure~\ref{fig:chordal-example}; we leave it to the reader to verify that they cannot be \textsc{mcs} end vertices.  
\begin{lemma}\label{lem:critical-edge}
  Let $z$ be a simplicial vertex of a connected chordal graph $G$.  If $N[z]$ is the end clique of a Prim ordering of $C(G)$, then all critical edges for $N[z]$ have the same label.
\end{lemma}
\begin{proof}
  Suppose for contradiction that there are two critical edges $e_1$ and $e_2$ for $N[z]$ with different labels.   For $i = 1, 2$, let $S_i$ be the label of $e_i$, and let ${\cal C}_i$ denote the set of components of $C(G) - S_i$ not containing $N[z]$.  We argue that for any $U_1\in {\cal C}_1$ and $U_2\in {\cal C}_2$, they are different and there is no edge between them.

  For $i = 1, 2$, by the definition of critical edges, there is a path from $N[z]$ to $e_i$; let $K_i$ denote the end of $e_i$ that is closer to $N[z]$ on this path.  There must be some clique $K'_i$ in $U_i$ containing $S_1$.
  Note that $K_i\cap K'_i = S_i$ because $K'_i$ and $K_i$ are in different components of $C(G) - S_i$.  Hence, $K_i K'_i$ is also a critical edge with label $S_i$ for $N[z]$.   There is a \stpath{N[z]}{K'_2} in $C(G) - S_1$, and hence $K'_2$ and $N[z]$ are connected in $C(G) - S_1$.  Likewise,  $K'_1$ and $N[z]$ are connected in $C(G) - S_2$.

  Since $S_1\ne S_2$ and they have the same cardinality, we can find $v_2\in S_2\setminus S_1\subset K'_2$.
  By Proposition~\ref{lem:separators-on-clique-graph}, $S_1$ is not a \stsep{z}{v_2}.  Thus, no maximal clique in $U_1$ contains $v_2$.  It follows that $U_1$ remains connected in $C(G) - S_2$ (note that $S_2$ is a minimum separator).  For the same reason, $U_2$ remains connected in $C(G) - S_1$.  If there exists an edge between $U_1$ and $U_2$, then this edge remains in at least one of $C(G) - S_1$ and $C(G) - S_2$:  It cannot have both labels $S_1$ and $S_2$.  But then $U_1$ and $U_2$ are connected in $C(G) - S_1$ or $C(G) - S_2$, neither of which is possible.

We can thus conclude that components in ${\cal C}_1\cup {\cal C}_2$ are disjoint and there is no edge among them.
  
    Let $\pi$ be a Prim ordering of $C(G)$ ended with $N[z]$.  Assume without loss of generality that the first visited clique in these components is from $U_1\in {\cal C}_1$, then we show that $N[z]$ is visited before all components $U_2\in {\cal C}_2$.  Since there is no edge between $U_1$ and $U_2$, before visiting $U_2$, it must visit a clique from the $z$-component of $C(G) - S_1$.  After that, however, it will not visit any edge of label $S_2$ before finishing this component.  Therefore $N[z]$ cannot be the end clique, a contradiction.  This concludes the proof.
\end{proof}

In other words, if $z$ is an \textsc{mcs} end vertex, then there is a unique minimum separator of $G$ that is ``closest to $z$'' in a sense.  This, although not sufficient, can be extended to a sufficient condition for \textsc{mcs} end vertices as follows.
To decide whether a simplicial vertex $z$ is an \textsc{mcs} end vertex, we can find the minimum separator $S$ in Proposition~\ref{cor:critical-edge} and focus on how the $z$-component of $C(G) - S$ is explored.  We have to start from a maximal clique not in it, and after that visit all maximal cliques in other components of $C(G) - S$ before the $z$-component.  In this juncture we may view the $z$-component as a separate graph and find all critical edges for $N[z]$ with respect to this component.  They also need to have the same label; suppose it is $S'$, which is strictly larger than $S$.  But this is not sufficient because we need to make sure that when $S$ is crossed, it can reach a maximal clique not in the $z$-component of $C(G) - S'$.
In Figure~\ref{fig:chordal-example}, if we delete vertices $v_{16}$ and $v_{17}$, (hence $K_9$,) then $K_6 K_7$ is the only critical edge for $K_8$.  The condition of Lemma~\ref{lem:critical-edge} is vacuously satisfied, but $v_{14}$ is still not  an \textsc{mcs} end vertex.  (Now $v_{18}$ is.)

Repeating this step recursively, we should obtain a sequence of separators with  increasing cardinalities.  Note that we only need to keep track of how these separators are crossed, while the ordering in each layer is irrelevant.  
This observation leads us to the following characterization, which subsumes Theorem~13 of Beisegel et al.~\cite{beisegel-18-end-vertex}.
For example, the sequence of critical edges for $N[v_1]$ in Figure~\ref{fig:chordal-example} are $K_6 K_7$, $K_2 K_5$, and $K_1 K_2$, which correspond to minimal separators $\{v_{11}\}$, $\{v_5, v_6\}$, and $\{v_2, v_3, v_4\}$, respectively.

\begin{theorem}\label{thm:mcs-characterization}
  Let $z$ be a simplicial vertex of a connected chordal graph $G$.  The clique $N[z]$ is a Prim end clique if and only if there is a sequence of edges $e_1$, $e_2$, $\ldots$, $e_k$ in $C(G)$, where the label of $e_i$ is $S_i$, on a path ended with $N[z]$ such that 
  \begin{enumerate}[(i)]
  \item $S_1$ is the label of critical edges for $N[z]$ and $S_k$ is the set of non-simplicial vertices in $N[z]$; and
  \item  for $1\le i < k$, in the $z$-component of $C(G) -S_i$, all the critical edges for $N[z]$ have the same lable, which is $S_{i+1}$.
  \end{enumerate}
  Moreover, every clique not in the $z$-component of $C(G) - S_1$ can be the start clique.
\end{theorem}
\begin{proof}
  We first show the if direction.  We may denote the two ends of $e_i$ by $K_i$ and $K'_i$, where $K'_i$ is in the $z$-component of $C(G) - S_i$.  (It is possible that $K'_i = K_{i+1}$ for some $1\le i < k$.)  For each $1\le i\le k$, we visit all the other components of $C(G) - S_i$ before using the edge $K_i K'_i$ to enter the  $z$-component, visiting $K'_i$.  
  This is possible because of Proposition~\ref{cor:critical-edge}, and as such we produce a Prim ordering of $C(G)$ that ends with $N[z]$.

  Now consider the only if direction, for which we construct the stated path by induction: We find the edges $e_1$, $e_2$, $\ldots$, $e_k$ in order, and show that for each $1\le i \le k$, the first $i$ edges can be extended to a path that ends with $N[z]$ and satisfies both conditions.  The first edge $e_1$ can be any critical edge for $N[z]$, and it is on a path ended with $N[z]$ because $C(G)$ is connected.  Now suppose that the first $i$ edges, namely, $e_1$, $\ldots$, $e_i$, have been selected, and we find $e_{i + 1}$ as follows.  For each $1\le j \le i$, let $T_j$ denote the $z$-component of $T_{j-1} - S_j$, where $T_0 = C(G)$.  If $T_i$ comprises the only maximal clique $N[z]$, we are done.

  Containing $N[z]$, cliques in $T_i$ are last visited by Proposition~\ref{cor:critical-edge}.  It is also a Prim ordering of the component itself.  Therefore, Lemma~\ref{lem:critical-edge} applies, and all the critical edges for $N[z]$ in $T_i$ have the same label.   Let $S_{i + 1}$ be this label, and let $T_{i+1}$ be the $z$-component of $T_i - S_{i+1}$.  We argue that there must be a maximal clique $K$ in $T_i - T_{i + 1}$ containing $S_i$; otherwise, the first component visited in $T_i - S_{i+1}$ would be the $z$-component, and then $N[z]$ cannot be the last visited clique.  We can use edge $K_i K$ to replace $e_i$,---note that they have the same label,---and choose any edge between $K$ and $N[z]$ with label $S_{i + 1}$ as $e_{i+1}$.  This concludes the inductive step and the proof.
\end{proof}

The proof of the only if direction of Theorem~\ref{thm:mcs-characterization} can be directly translated into an algorithm to decide Prim end cliques, implying a polynomial-time algorithm for the \textsc{mcs} end vertex problem on chordal graphs.  This algorithm however has to take $\Omega(n^2)$ time because the size of $C(G)$.
We show a very simple algorithm below, which itself best reveals the spirit of graph searches.  As long as we cross the separators in the order specified in Theorem~\ref{thm:mcs-characterization}, and make sure we finish other components before visiting the $z$-component, then it is the Prim ordering we need.  On the other hand, a run of Prim's algorithm started from $N[z]$ will cross the separators in the reversed order, and before crossing the $i$th separator $S_i$, it has to exhaust the whole $z$-component $C(G) - S_i$.  
\begin{figure}[h!]
  \tikz\path (0,0) node[draw=gray!50, text width=.9\textwidth, rectangle, rounded corners, inner xsep=20pt, inner ysep=10pt]{
    \begin{minipage}[t!]{\textwidth} \small
  {\sc Input}: A graph $G$ and an \textsc{mcs} ordering $\sigma$ of $G$ started with $z$.
  \\
  {\sc Output}: Whether $z$ can be an \textsc{mcs} end vertex of $G$.
  \begin{tabbing}
    AAA\=AAa\=AAA\=AAA\=MMMMMMMMMMMMMAAAAAAAAAAAAAAAAAAAAAAAAA\=A \kill
    1.\> {\bf for} $i \leftarrow 1$ to $n$ {\bf do}
    \\
    1.1.\>\> $D\leftarrow $ the set of unvisited vertices with the maximum number of visited neighbors;
    \\
    1.2.\>\> visit the vertex $\argmax_{v \in D}\sigma(v)$;
    \\
    2.\> {\bf if} the last visited vertex is $z$ {\bf then return} ``yes'';
    \\
    \> {\bf else return} ``no.''
  \end{tabbing}
    \end{minipage}
  };
  \caption{Algorithm for deciding whether a vertex $z$ is an \textsc{mcs} end vertex of a chordal graph.}
  \label{fig:alg-mcs-chordal}
\end{figure}

\begin{proof}[Proof of Theorem~\ref{thm:mcs-chordal}]
  Let $G$ be a connected chordal graph.  We find an \textsc{mcs} ordering $\sigma$ of $G$ started with $z$, and then use the algorithm descibed in Figure~\ref{fig:alg-mcs-chordal}.  We first show its correctness: Vertex $z$ is an \textsc{mcs} end-vertex of $G$ if and only if $z$ is the last visited vertex.  The if direction is correct because the algorithm conducts \textsc{mcs}, and Hence we focus on the only if direction.
  Let $S_1$, $\ldots$, $S_k$ be the set of separators specified in Theorem~\ref{thm:mcs-characterization}, and let $\sigma^+$ denote the ordering returned by the algorithm in Figure~\ref{fig:alg-mcs-chordal}.  We show by induction that for each $1 \le i\le k$, vertices in all the other components of $G - S_i$ are visited before those in the same component with $z$.

  Let $T'_1$ be the component of $G - S_1$ containing $z$.  By Proposition~\ref{cor:critical-edge} and Corollary~\ref{cor:weighted-clique-tree}, vertices in $T_1$ are at the beginning of $\sigma$.  In each component of $G - S_1$, there is a vertex adjacent to all vertices in $S_1$.  When the first vertex in $T'_1$ is being visited, it has precisely $|S_1|$ visited neighbors, i.e., $S_1$.  By the selection of vertices in step~1, all other components have been finished.  Thus, $T'_1$ is the last visited component of $G - S_1$.

  For the inductive step, suppose that the induction hypothesis is true for all $p$ with $1\le i\le p< k$, we show it is also true for $p + 1$.  For $1< i \le k$, let  $T'_i$ be the component of $T_{i-1} - S_i$ containing $z$, and let $T_i$ be the subgraph induced by $V(T'_i)\cup S_i$.  Let $v\in T_{p+1}$ be the vertex satisfying $v<_{\sigma^+} u$ for all $u\in T_{p+1}\setminus \{v\}$.  Then $S_{p+1}\subseteq N(v)$ and $x<_{\sigma^+} v$ for all $x\in S_{p+1}$.  Since $S_{p+1}$ is a minimum separator of $T_p$, any other component of $T_p - S_{p + 1}$ has a vertex adjacent to all of $S_{p+1}$.  Such a vertex $x$ would satisfy $v <_\sigma x$ because of  Proposition~\ref{cor:critical-edge} and Corollary~\ref{cor:weighted-clique-tree}, and then be chosen by step~1 before $v$.  Now that all the vertices in $G - N[z]$ and the non-simplicial vertices in $N[z]$ have been visited, the only remaining vertices are true twins of $z$.  Since $\sigma(z) = 1$, it has to be the last visited.  This concludes the proof of the correctness.

  We now analyze the running time.  The only difference between the algorithm and the original \textsc{mcs} algorithm is step~1.2.  We need to compare the $\sigma$-numbers of vertices in $D$.  It needs to be done $n$ times, and each time takes $O(n)$ time, and hence the extra time is $O(n^2)$.  Together with the time for \textsc{mcs} itself, the total running time is $O(n^2 + m)= O(n^2)$.
\end{proof}

This algorithm can be called the \textsc{mcs$^+$} algorithm.  Unlike \textsc{lbfs$^+$}~\cite{corneil-04-survey-lbfs}, however, it is not immediately clear how to carry \textsc{mcs$^+$} out in linear time.

\section{Maximum cardinality search on weakly chordal graphs}
A graph $G$ is \emph{weakly chordal} if neither $G$ nor its complement contains an induced cycle on five or more vertices.  Since the complement of each induced cycle on six or more vertices contains an induced cycle on four vertices, all chordal graphs are weakly chordal.
To prove the NP-completeness of the \textsc{mcs} end vertex problem on weakly chordal graphs, we use a reduction from the 3-satisfiability problem (3-\textsc{sat}), in which each clause comprises precisely three literals.

Given an instance $\mathcal{I} $ of 3-\textsc{sat} with $p$ variables and $q$ literals, we construct a graph $G$ as follows (see Figure~\ref{fig:weaklychordal} for an example).  Let the variables and clauses of  $\mathcal{I}$ be denoted by $x_1$, $x_2$, $\ldots$, $x_p$ and $c_1$, $c_2$, $\ldots$, $c_q$, respectively.
For each literal, (including those that do not occur in any clause,) we introduce a vertex; let $L$ denote this set of $2 p$ literal vertices.  For each literal vertex, we add edges between it and other vertices in $L$, with the only exception of its negation.  We also introduce a set $C$ of $q$ clause vertices, each for a different clause; they forms an independent set.  For each $\ell\in L$ and $c\in C$, we add an edge $\ell c$ if the literal $\ell$ does not occur in the clause $c$.  Therefore, each clause vertex has $2p - 3$ neighbors in $L$.  Finally, we add seven extra vertices $a_1, a_2, u_1, u_2, b, y, z$ and edges $a_1 a_2$, $u_1 u_2$, $y z$, $\{b, z\}\times L$ and $\{a_2, u_1, u_2, y\}\times (L\cup C)$.

\begin{figure}[h]
  \centering      \footnotesize
    \begin{tikzpicture}[every node/.style={vertex}, scale=1.5]

    \draw[rounded corners=3mm] (-0.5,-4) rectangle (5.3,0.1);
    \filldraw[fill=gray!15,rounded corners=3mm] (-0.3,-1.5) rectangle (5.7,-0.1);

    \node[label=right:$a_2$] (S) at (2.4,0.8){};
    \node[label=left:$a_1$] (SS) at (1.7,0.8){};
    \node[label=left:$u_1$] (U1) at (3.7,0.8){};
    \node[label=right:$u_2$] (U2) at (4.4,0.8){};
    \node[label=left:$x_{1}$] (X1) at (0.3,-0.4) {};
    \node[label=left:$x_{2}$] (X2) at (1.7,-0.4) {};
    \node[label=right:$x_{3}$] (X3) at (3.1,-0.4) {};
    \node[label=right:$x_{4}$] (X4) at (4.5,-0.4) {};
    \node[label=left:$\overline {x_{1}}$] (X11) at (0.3,-1.2) {};
    \node[label=left:$\overline {x_{2}}$] (X22) at (1.7,-1.2) {};
    \node[label=right:$\overline {x_{3}}$] (X33) at (3.1,-1.2) {};
    \node[label=right:$\overline {x_{4}}$] (X44) at (4.5,-1.2) {};
    \node[label=right:$b$] (M) at (6.5,-0.45){};
    \node[label=right:$z$] (Z) at (6.5,-1.1){};
    \node[label=right:$y$] (Y) at (6.5,-2.8){};
    \node[label=below:$\overline {x_1} \vee \overline x_2 \vee \overline {x_3}$] (XX1) at (0.3,-3){};
    \node[label=below:${x_1} \vee \overline {x_2} \vee x_4$] (XX2) at (2.25,-3){};
    \node[label=below:$\overline {x_2} \vee \overline {x_3} \vee \overline {x_4}$] (XX3) at (4.5,-3){};
    
    \begin{scope}[every path/.style={ultra thin}]
    \draw[line width=0.6,dashed]  
    (X1)--(X11) (X2)--(X22) (X3)--(X33) (X4)--(X44)
    (XX1)--(X11) (XX1)--(X22) (XX1)--(X33)
    (XX2)--(X1) (XX2)--(X22) (XX2)--(X4)
    (XX3)--(X22) (XX3)--(X33) (XX3)--(X44)
    (XX1)--(XX2) (XX2)--(XX3);
    \draw[line width=0.5,dashed] (XX1)..controls(2.25,-3.4)..(XX3);
    \draw[line width=0.5]
    (S)--(SS) (S)--(2.25,0.1) (S)--(2.35,0.1)  (S)--(2.45,0.1) (S)--(2.55,0.1) 
    (U1)--(U2)
    (U1)--(3.55,0.1) (U1)--(3.65,0.1) (U1)--(3.75,0.1) (U1)--(3.85,0.1)
    (U2)--(4.25,0.1) (U2)--(4.35,0.1) (U2)--(4.45,0.1) (U2)--(4.55,0.1)
    (M)--(5.7,-0.3) (M)--(5.7,-0.4) (M)--(5.7,-0.5) (M)--(5.7,-0.6)
    (Z)--(5.7,-0.95) (Z)--(5.7,-1.05) (Z)--(5.7,-1.15) (Z)--(5.7,-1.25)
    (Y)--(Z) (Y)--(5.3,-2.65) (Y)--(5.3,-2.75) (Y)--(5.3,-2.85) (Y)--(5.3,-2.95);

  \end{scope}
  \end{tikzpicture} 

  \caption{Construction for NP-completeness proof of the \textsc{mcs} end vertex problem on weakly chordal graphs.   The \textsc{3-sat} instance has four variables and three clauses, $(\overline {x_1} \vee \overline x_2 \vee \overline {x_3})$, $({x_1} \vee \overline {x_2} \vee x_4)$, $(\overline {x_2} \vee \overline {x_3} \vee \overline {x_4})$, i.e., $p = 4$ and $q = 3$.  The $2p$  literal vertices are shown in the small gray box, and the $q$ clause vertices are in the big box.   In the boxes, two vertices are nonadjacent if there is a dashed line between them, and adjacent otherwise.  Vertices $b$ and $z$ are adjacent to all literal vertices, while vertices $a_2, u_1, u_2$, and $y$ are adjacent to all literal vertices and all clause vertices.
The \textsc{mcs} ordering $\langle a_1, a_2, x_1, \overline{x_2}, x_3, x_4, b, \overline{x_1}, x_2, \overline{x_3}, \overline{x_4}, u_1, u_2, y, c_1, c_2, c_3, z\rangle$ of $G$ corresponds to the satisfying assignment in which all variables but $x_2$ are set to be true.
  }
\label{fig:weaklychordal}
\end{figure}

\begin{proposition}\label{prop:weaklychordalgraph}
  The graph $G$ constructed above is a weakly chordal graph.
\end{proposition}
\begin{proof}
  We need to show that neither $G$ nor $\overline {G}$ contains an induced cycle on five or more vertices.
  We proceed as follows: We identify a vertex $v\in V(G)$ such that $G$ contains an induced cycle on five or more vertices if and only if $G-v$ contains an induced cycle on five or more vertices, and then consider $G-v$.
  The following properties are straightforward:

  \begin{enumerate}[(i)]
  \item A vertex on any induced cycle on five or more vertices has degree at least two.
  \item A simplicial vertex is not on any induced cycle on five or more vertices.
  \item An induced cycle on five or more vertices cannot contain a pair of true twins or false twins, and when it contains one of them, this vertex can be replaced by the other.
  \item If a vertex is on an induced cycle on five or more vertices, then it has at least two non-neighbors, and there is at least one edge among these non-neighbors.
  \end{enumerate}

  We can reduce $G$ to $G - \{a_1\}$ because $d(a_1) = 1$ and (i); then to $G - \{a_1, u_2\}$ because $u_1$ and $u_2$ are true twins and (iii); to $G - \{a_1, u_1, u_2\}$ because $u_1$ and $a_2$ are false twins in $G - \{a_1, u_2\}$ and (iii); to $G - \{a_1, u_1, u_2, y\}$ because the only two remaining non-neighbors of $y$, namely,  $a_2$ and $b$, are not adjacent to each other and (iv); to $G - \{a_1, u_1, u_2, y, a_2\}$ for the same reason; to $G - \{a_1, u_1, u_2, y, a_2, b\}$ because $z$ and $b$ are false twins in $G - \{a_1, u_1, u_2, y, a_2\}$ and (iii); and finally to $G - \{a_1, u_1, u_2, y, a_2, b, z\}$ because the only non-neighbors of $z$, namely, $C$, are independent and (iv).
  The remaining graph is $G[L\cup C]$.  Suppose that there is an induced cycle $H$ on five or more vertices.  It must intersect both $L$ and $C$, since each vertex in $L$ has only one non-neighbor in it, and since $C$ is independent.  Let $v\in C$ be a vertex on this cycle.  Its two neighbors on $H$ have to be from $L$; and since they are nonadjacent to each other, they have to be $x$ and $\bar x$ for some variable $x$.  Since both $x$ and $\bar x$ are adjacent to all other vertices in $L$, the other $\ge 2$ vertices on $H$ have to be from $C$.  But this is impossible because $C$ is independent.

  Now we consider $\overline {G}$.
  It can be reduced to $\overline G - \{a_1\}$ because $a_1$ has only one non-neighbor and (iv); then to $\overline G - \{a_1, u_2\}$ because $u_1$ and $u_2$ are false twins and (iii); to $\overline G - \{a_1, u_1, u_2\}$ because $u_1$ and $a_2$ are true twins in $\overline G - \{a_1, u_2\}$ and (iii); to $\overline G - \{a_1, u_1, u_2, y\}$ because $y$ is simplicial in $\overline G - \{a_1, u_1, u_2\}$ and (ii); to $\overline G - \{a_1, u_1, u_2, y, b\}$ because $z$ and $b$ are true twins in $\overline G - \{a_1, u_1, u_2, y\}$ and (iii); to $\overline G - \{a_1, u_1, u_2, y, b, a_2\}$ because the degree of $a_2$ is one in $\overline G - \{a_1, u_1, u_2, y, b\}$ and (i); and finally to $\overline G - \{a_1, u_1, u_2, y, a_2, b, z\}$ because $z$ is simplicial in  $\overline G - \{a_1, u_1, u_2, y, b, a_2\}$ and (ii).
  The remaining graph is $\overline G[L\cup C]$.  Suppose that there is an induced cycle $H$ on five or more vertices.  Since $C$ is a clique, $H$ contains at most two vertices from $C$.  In other words, at least three vertices on $H$ are from $L$, but this is impossible because each vertex in $L$ has only one neighbor in $L$.

We can thus conclude that $G$ is a weakly chordal graph.
\end{proof}

We are now ready to prove Theorem~\ref{thm:mcs-weakly-chordal}.
\begin{proof}[Proof of Theorem~\ref{thm:mcs-weakly-chordal}]
  It is clear that the \textsc{mcs} end vertex problem is in NP, and we now show that it is NP-hard.
Let $\mathcal I$ be an instance of 3-\textsc{sat}, and let $G$ be the graph constructed from $\mathcal I$.  We show that $z$ is an \textsc{mcs} end-vertex of $G$ if and only if $\mathcal I$ has a satisfying assignment.
  
For the if direction, suppose that $\cal I$ is satisfiable, and we give an \textsc{mcs} ordering $\sigma$ as follows.  Let us fix a satisfying assignment of $\mathcal I$, and let $T$ be the set of variables that are set to be true.   The starting vertex is $a_1$, which is followed by $a_2$; visited after them are  $\{x\mid x\in T\} \cup \{\bar x\mid x\not\in T\}$, (i.e., the literal vertices corresponding to true literals,) in any order.  After these $p + 2$ vertices, each of $y, z, u_1, u_2$, $b$, and each of the unvisited literal vertices has $p$ visited neighbors.  On the other hand, each clause vertex has at most $p$ visited neighbors: Each clause contains a true literal, and hence each clause vertex has at least one non-neighbor in the visited literal vertices.

Then $\sigma(b) = (p + 3)$.  Since $b$ is adjacent to only literal vertices, the next vertex is one of them.   On the other hand, since vertices $L\setminus T$ form a clique, they have to be visited between $p + 4$ and $2p + 3$, i.e., before others.

The remaining vertices are $u_1$, $u_2$, $y$, $z$, and clause vertices.  Each of $u_1$, $u_2$, $y$, and $z$ has $2 p$ visited neighbors, while each clause vertex has only $2 p - 2$, because each clause is nonadjacent to three literal vertices.  Let $u_1$, $u_2$, and $y$ be visited next.  After that, all the remaining vertices ($z$ and all clause vertices) have the same number of visited neighbors, $2 p + 1$.  There is no edge among these vertices, so they an be visited in any order.  We have thus obtained an \textsc{mcs} ordering of $G$ ended with $z$.

We now prove the only if direction.  Suppose that $\sigma$ is an \textsc{mcs} ordering of $G$ with $\sigma(z) = n$.
Since $N(z) = N(b)\cup\{y\}$, visiting $y$ before $b$ would force $z$ to be visited before $b$; therefore, $b <_\sigma y <_\sigma z$.  Likewise, $N(b) = L \subset L \cup C \subset N(y)$ and $b <_\sigma y$ demand
\begin{equation}
  \label{eq:1}\tag{$\star$}
  b <_\sigma c \text{ for all } c \in C.
\end{equation}
Since $d(a_1) = 1$, it is easy to verify that $\{\sigma(a_1), \sigma(a_2)\} = \{1, 2\}$; otherwise, $\sigma$ must end with $a_1$.  The third vertex of $\sigma$ has to be from $N(a_2)$, i.e., $L\cup C$.  It cannot be from $C$ because of \eqref{eq:1}.  Therefore, $X = \{\ell\mid 3\le \sigma(\ell)\le p + 2\}\subset L$: (1) For each variable, one literal vertex has more visited neighbors than $b$, $z$, $y$, $u_1$, $u_2$; (2) clause vertices cannot be visited before $b$.  There cannot be any variable $x$ such that both $x,\bar x\in X$, because $x \bar x\not\in E(G)$.  We claim that assigning a variable $x$ to be true if and only if $x\in X$ is a satisfying assignment for $\cal I$.
Suppose for contradiction that some clause $c$ is not satisfied by this assignment.
By the construction of $G$, the clause vertex $c$ is adjacent to all vertices of $X$.  After visiting the first $p + 2$ vertices, $c$ has $p + 1$ visited neighbors, ($\{a_2\}\cup X$,) while any other unvisited vertex in $V(G)\setminus C$ has at most $p$ visited neighbors.  But then $\sigma(c) = k + 3$, contradicting \eqref{eq:1}.  Therefore, all clauses are satisfied, and this completes the proof.
\end{proof}

\section{Lexicographic depth-first search on chordal graphs}
Berry et al.~\cite[Characterization 8.1]{berry-10-extremities-search} have given a full characterization of \textsc{mns} end vertices on chordal graphs: A vertex $z$ is an \textsc{mns} end vertex if and only if it is simplicial and the minimal separators of $G$ in $N(z)$ are totally ordered by inclusion.  Since \textsc{ldfs} is a special case of \textsc{mns}, its end vertices also have this property.  We show that this condition is also sufficient for a vertex to be an \textsc{ldfs} end vertex.

Similar as \textsc{dfs}, \textsc{ldfs} visits a neighbor of the most recent vertex, or backtracks if all its neighbors have been visited.  The difference lies on the choice when the vertex has more than one unvisited neighbors.  Each unvisited vertex has a label, which is all its visited neighbors.  When there are ties, it chooses a vertex with the lexicographically largest label.  The following is actually a simple property of \textsc{dfs}. 
\begin{proposition}
  \label{lem:ldfs-component}
  Let $X\subseteq V(G)$ such that $G[X]$ is connected.  If an \textsc{ldfs} visits all vertices in $N(X)$ before the first vertex in $X$, then it visits vertices in $X$ consecutively.
\end{proposition}

\begin{lemma}\label{lem:ldfs-chordal}
 A vertex $z$ of a chordal graph $G$ is an \textsc{ldfs} end vertex if and only if it is simplicial and the minimal separators of $G$ in $N(z)$ are totally ordered by inclusion.
\end{lemma}
\begin{proof}
  The only if direction follows from that all \textsc{ldfs} orderings are \textsc{mns} orderings \cite{corneil-08-graph-searching} and the result of Berry et al.~\cite{berry-10-extremities-search}.
  For the if direction, suppose that $S_1$, $\ldots$, $S_k$ are the minimal separators in $N(z)$ and $S_1 \subset \cdots \subset S_k$.  It is easy to see that for all $1\le i\le k$, each component of $G - S_i$ not containing $z$ is a component of $G - S_k$; let $\cal C$ denote these components.  We show an \textsc{ldfs} ordering $\sigma$ of $G$ as follows.  It starts from visiting all vertices in $S_1$, followed by components $C\in \cal C$ with $N(C) = S_1$, visited one by one.  In the same manner, it deals with $S_2. \ldots S_k$ in order.  After that the only unvisited vertex are $z$ and its true twins, of which it chooses $z$ the last.  We now verify that this is indeed a valid \textsc{ldfs} ordering.  It is clear for $S_1$.  Since vertices in each component $C\in \cal C$ are visited after $N(C)$, By Proposition~\ref{lem:ldfs-component}, it suffices to show the correctness when it visits a vertex in $N(z)$ and when it visits the first vertex of a new component $C\in \cal C$.  When such a decision is made, the label of an unvisited vertex is either $\emptyset$ or all visited vertices in $N(z)$, i.e., the most recently visited separator.  So it is always correct to select a vertex from $N(z)$.  When a vertex $v$ in a component $C$ is selected, the visited vertices in $N(z)$ are precisely $N(C)$, hence $v$ does have the largest label.
\end{proof}

\section{Breadth-first search on interval graphs}

Interval graphs are intersection graphs of intervals on the real line.  An interval graph is always chordal, and in particular, it has a clique tree that is a path \cite{fulkerson-65-interval-graphs}.
Corneil et al.~\cite{corneil-10-end-vertices-lbfs} gave a very simple linear-time algorithm for deciding whether a vertex $z$ is an \textsc{lbfs} end vertex of an interval graph, which is very similar to our algorithm in Figure~\ref{fig:alg-mcs-chordal}.  They conducted an \textsc{lbfs} started from $z$, and then another \textsc{lbfs} that uses the first run to break ties.  They proved that $z$ is an \textsc{lbfs} end vertex if and only if it is the last of the second run.
As shown in Figure~\ref{fig:bfs-bad-example}, however, this algorithm cannot be directly adapted to the \textsc{bfs} end vertex problem.

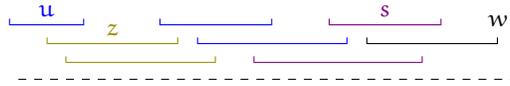
\begin{figure}[h]
  \centering      \footnotesize
    \begin{tikzpicture}[scale=1.]

      \begin{scope}[shift = {(0,  0)}, scale=.25]
      \draw[dashed,thin] (3.5,-0.9) -- (30,-0.9);
      \begin{scope}[every path/.style={{|[right]}-{|[left]}}]
        \draw[olive] (5, 1) to["$z$"] (12, 1);
      \draw[blue] (3, 2) to["$u$"] (7, 2);
      \draw[olive] (6, 0) to (14, 0);  
      \draw[blue] (13, 1)  to  (21,1); 
      \draw[blue] (11,2) to (17,2);
      \draw[violet] (20, 2) to ["$s$"] (26, 2) ;
      \draw[violet] (16, 0) node {}  to (25, 0);
      \draw  (22,1)  to (29,1) node["$w$"] {};
      \end{scope}
    \end{scope}
  \end{tikzpicture} 
  \caption{A \textsc{bfs} started from $z$ may end with $s$ or $w$, but a \textsc{bfs} started from $w$ has to end with $u$.  (Note that a \textsc{bfs} started from $s$ may end with $z$.)}
  \label{fig:bfs-bad-example}
\end{figure}

If a graph has one and only one universal vertex, then each of the other vertices is a BFS end-vertex, but not itself.  If it has two or more universal vertices, then every vertex can be a BFS end-vertex.  Therefore, we may focus on graphs with no universal vertex.  Such an interval graph has at least three maximal cliques. 

\begin{proposition}[\cite{corneil-09-lbfs-strucuture-and-interval-recognition}]
  \label{lem:interval-ends}
  Let $G$ be a connected interval graph, and let $K_1,\ldots,K_p$ be a clique path of $G$.  Let $u\in K_1$ and $w\in K_p$ be two simplicial vertices.
  \begin{enumerate}[(i)]
  \item Both $u$ and $w$ are \textsc{lbfs} end vertices.
  \item For any vertex $v\in V(G)$, one of $u$ and $w$ has the largest distance to $v$.
  \end{enumerate}
\end{proposition}

It is known that a vertex $z$ of an interval graph $G$ can be an \textsc{lbfs} end vertex if and only if it is simplicial and $N[z]$ can be one of the two ends of a clique path of $G$ \cite{corneil-09-lbfs-strucuture-and-interval-recognition}.  However, a \textsc{bfs} may satisfy neither of the two conditions.  In Figure~\ref{fig:bfs-bad-example}, for example, vertex $z$ is not simplicial but can be a \textsc{bfs} end vertex.  When $z$ is not in an end clique, it should be close to one.  Actually, it should be at distance at most two to one of the $u$ and $w$ as specified in Proposition~\ref{lem:interval-ends}.  However, a \textsc{bfs} end vertex might be at distance two to both $u$ and $w$, as shown in Figure~\ref{fig:bfs-example}.

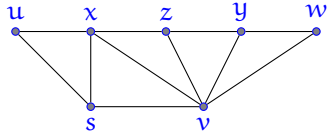
\begin{figure}[h]
  \centering\small
  \begin{tikzpicture}[every node/.style={vertex}]
    \foreach[count = \i from 0] \v in {u, x, z, y, w} 
    \node["$\v$"] (\v) at (\i, 1.) {};
    \node["$s$" below] (s) at (1,0) {};
    \node["$v$" below] (v) at (2.5,0) {};
    \draw (s) -- (u) -- (x) -- (z) -- (y) -- (w) -- (v) -- (y);
    \draw (z) -- (v) -- (x) -- (s) -- (v);
  \end{tikzpicture}
    \caption{$s, v, u, x, w, y, z$ is a \textsc{bfs} ordering ended with $z$.}
  \label{fig:bfs-example}
\end{figure}

For a fixed clique path $K_1,\ldots,K_p$ of an interval graph $G$, we let $\lp{v}$ and $\rp{v}$ denote, respectively, the smallest  and the largest number $i$ such that $v\in K_i$.\footnote{One may note that $\{v: [\lp{v}, \rp{v}]\}$ gives an interval representation for $G$.}  We use $\mathrm{dist}(u, v)$ to denote the distance between $u$ and $v$.

\begin{figure}[h!]
  \tikz\path (0,0) node[draw=gray!50, text width=.9\textwidth, rectangle, rounded corners, inner xsep=20pt, inner ysep=10pt]{
    \begin{minipage}[t!]{\textwidth} \small
      {\sc Input}: A connected interval graph $G$, a clique path $K_1, \ldots, K_p$ of $G$,
      \\ \phantom{{\sc Input}: }simplicial vertices $u\in K_1$ and $w\in K_p$, and $z\in V(G)$.
  \\
  {\sc Output}: Whether there exists a \textsc{bfs} ordering $\sigma$ of $G$ with $\sigma(z) = n$ and $u <_\sigma w$.
  \begin{tabbing}
    AAA\=AAa\=AAA\=AAA\=MMMMMMMMMMMMMAAAAAAAAAAAAAAAAAAAAAAAAA\=A \kill
    1.\> {\bf if} $z = w$ {\bf then return} ``yes'';
    \\
    2.\> {\bf if} there exists a universal vertex in $V(G)\setminus \{z\}$ {\bf then return} ``yes''; 
    \\
    3.\> $X\leftarrow \{x\in V(G): \mathrm{dist}(x, z) = \mathrm{dist}(x, w) \ge \mathrm{dist}(x, u)\}$;
    \\
    4.\> {\bf if} $X = \emptyset$ {\bf then return} ``no''; 
    \\
    5.\> $s\leftarrow$ any vertex in $\argmin_{v\in X} \lp{v}$;
    \\
    6.\> {\bf if} $\rp{z} < \lp{s}$ {\bf then return} ``no'';
    \\
    7.\> {\bf if} $s = u$ {\bf then return} ``yes'';
    \\
    8.\> {\bf for each} vertex $v\in N(s)$ at distance $\mathrm{dist}(s, u) - 1$ to $u$ {\bf do}
    \\
    \>\> {\bf if} $\mathrm{dist}(v, z) > \mathrm{dist}(v, u)$ {\bf then return} ``yes'';
    \\
    9.\> {\bf return} ``no.'' 
  \end{tabbing}
    \end{minipage}
  };
\caption{Main procedure for \textsc{bfs} end vertex on interval graphs.}
\label{fig:alg-BFS-interval}
\end{figure}

\begin{lemma}\label{lemma:bfs-endvertex-interval}
  The \textsc{bfs} end vertex problem can be solved in $O(n+m)$ time on interval graphs.
\end{lemma}
\begin{proof}
  Let $G$ be an interval graph; we may assume without loss of generality that $G$ is connected.  We use the algorithm of Corneil et al.~\cite{corneil-09-lbfs-strucuture-and-interval-recognition} to build a clique path for $G$, and take simplicial vertices $v_1, v_2$ from the first and last cliques of the clique path.  We call the procedure described in Figure~\ref{fig:alg-BFS-interval} twice, first with $u = v_1, w = v_2$; in the second call, we reverse the clique path, and use $u = v_2, w = v_1$.  Suppose that the procedure is correct, then vertex $z$ is a \textsc{bfs} end vertex if and only if at least one of the two calls returns yes.  In the rest we prove the correctness of the procedure and analyze its running time.

  We start from characterizing the first vertex $s$ of a \textsc{bfs} ordering $\sigma$ with $\sigma(z) = n$ and $u <_\sigma w$, if one exists.  Since $u <_\sigma w <_\sigma z$, we must have $\mathrm{dist}(s, u) \le \mathrm{dist}(s, w) \le \mathrm{dist}(s, z)$.  On the other hand, Proposition~\ref{lem:interval-ends} implies $\mathrm{dist}(s, z) \le \max\{\mathrm{dist}(s, u), \mathrm{dist}(s, w)\} = \mathrm{dist}(s, w)$.  Therefore, a desired \textsc{bfs} ordering $\sigma$, if it exists, must start from a vertex $s$ satisfying
  \begin{equation}
    \label{eq:start-vertex}\tag{$\dag$}
    \mathrm{dist}(s, z) = \mathrm{dist}(s, w) \ge \mathrm{dist}(s, u).
  \end{equation}
  We argue that at least one of the following is true for $z$:
  \begin{itemize}
  \item on any shortest \stpath{s}{u}, $z$ is adjacent to the second to last vertex but no vertex before it.
  \item on any shortest \stpath{s}{w}, $z$ is adjacent to the second to last vertex but no vertex before it.
  \end{itemize}
  Let $P_u$ be any \stpath{s}{u} and $P_w$ any \stpath{s}{w}.  Since they together form a \stpath{u}{w} that visits all the maximal cliques of $G$, vertex $z$ is adjacent to at least one of these two paths.  If $z$ is adjacent to a vertex on $P_u$, then it has to be the last two; otherwise $\mathrm{dist}(s, z) < \mathrm{dist}(s, u)$.  Since $u$ is simplicial, $z$ is adjacent to its neighbor on the path if $z u \in E(G)$.  Therefore, $z$ is always adjacent to the second to last vertex on this path.  The same argument applies if $z$ is adjacent to $P_w$.
  
  The correctness of step 1 follows from Proposition~\ref{lem:interval-ends}.  For step~2, note that if $v\ne z$ is a universal vertex, then $\langle v, u, w, \ldots, z \rangle$ is such a \textsc{bfs} ordering.  Steps 3 and 4 are justified by \eqref{eq:start-vertex}.  When the algorithm reaches step~5, $X$ is not empty, and hence $s$ is well defined.
  Let $q = \mathrm{dist}(s, z) = \mathrm{dist}(s, w)$.  Note that $q \ge 2$ because $s$ is not universal.  Hence, $z, w \not \in N(s)$.  

  We show the correctness of step~6 by contradiction.  Suppose that $\rp{z} < \lp{s}$ but there exists a \textsc{bfs} ordering $\sigma$ with $\sigma(z) = n$ and $u <_\sigma w$.  Let $s'$ be the first vertex of $\sigma$.  Since $s' \in X$, the selection of $s$ implies $\lp{s}\le \lp{s'}$.  Then $\rp{u} = 1\le \rp{z} < \lp{s}\le \lp{s'}$, therefore, $\mathrm{dist}(s', u) \ge 2$.   In this case, on any shortest \stpath{s'}{u}, $z$ is adjacent to the second to last vertex but no vertex before it.  Hence, $\mathrm{dist}(s', z) = \mathrm{dist}(s', u) = \mathrm{dist}(s', w)$; let it be $q'$.  Since $u<_\sigma w$, there must be some neighbor $u''$ of $u$ at distance $q' - 1$ to $s'$ visited before neighbors of $w$.  The vertex $u''$ cannot be universal, hence nonadjacent to $w$.   But $u''$ is adjacent to $z$, which implies $z<_\sigma w$, a contradiction.  Therefore, step~6 is correct, which means $\rp{s} <\lp{z}$ because $s$ and $z$ are not adjacent.  
Let $s=w_0,w_1,\ldots,w_{q-1}, w_q = w$ be a shortest \stpath{s}{w}.  Note that $w_{q-1} \in N(z)$.

For step~7, it suffices to give the following \textsc{bfs} ordering, which starts with $s= u$.  Of all vertices at distance $i$ to $s$, $1\le i \le q$, the first visited vertex is $w_i$.  Note that every vertex is adjacent to $w_1, \ldots, w_{q-1}$.  From $\rp{w_{q-1}} = p$ it can be inferred that all vertices at distance $q$ to $s$ are adjacent to $w_{q-1}$.  Since $w_{q-1}$ is the first visited vertex at level $q - 1$, vertices at distance $q$ to $s$ can be visited in any order.  Therefore, we can have a \textsc{bfs} ordering $\sigma$ of $G$ with $u <_\sigma w$ and $\sigma(z) = n$ .

We now consider step 8, for which we show that there exists a \textsc{bfs} ordering $\sigma$ with $\sigma(s) = 1$, $\sigma(v) = 2$, $\sigma(z) = n$, and  $u<_\sigma w$.  Note that $\mathrm{dist}(w_1, z) = \mathrm{dist}(w_1, w) = q - 1$.  Therefore $v\ne w_1$; otherwise step~5 should have chosen $v$ because $\lp{v} < \lp{s}$.  For $1 \le i \le q-1$, vertex $w_i$ is always visited in the earliest possible time; in particular, $\sigma(w_1) = 3$.  Since $v$ is on a shortest \stpath{s}{u}, $u$ is a descendant of $v$ in the \textsc{bfs} tree generated by $\sigma$.  On the other hand, since both $\mathrm{dist}(v, z)$ and $\mathrm{dist}(v,w)$ are larger than $\mathrm{dist}(v,u)$, either vertices $z$ and $w$ are not descendants of $v$, or they are at a lower level than $u$.  In either case, we have $u <_\sigma w$.  When $w_{q}$ is visited, all the unvisited vertices are at distance $q$ to $s$ and adjacent to $w_{q-1}$.  Thus, we can have $\sigma(z) = n$.

We are now at the last step.  Note that the algorithm can reach here only when $\mathrm{dist}(s, z) = \mathrm{dist}(s,w) = \mathrm{dist}(s,u)$: The condition of step~8 must be true if $\mathrm{dist}(s,u) < q$.  Suppose for contradiction that there exists a \textsc{bfs} ordering $\sigma$ with $\sigma(z) = n$ and $u <_\sigma w$ but no vertex satisfies the condition in step~8.   Let $s'$ be the starting vertex of $\sigma$.
Since $s'\in X$ and by the selection of $s$, we have $\lp{s'} \ge \lp{s}$, which implies $\mathrm{dist}(s', u) \ge \mathrm{dist}(s, u)$.  Note that $s'$ is adjacent to any \stpath{s}{w}, and hence its distance to $w$ is at most $q + 1$.   In summary,

$$q= \mathrm{dist}(s, u) \le \mathrm{dist}(s', u) \le \mathrm{dist}(s', w) \le q + 1.$$
Let $Y$ denote all vertices at distance $q - 1$ to $u$, and let $Z$ denote all vertices at distance $q - 1$ to $w$.  Note that $Y$ is disjoint from $Z$: A vertex in $v\in Y\cap Z$ would be adjacent to $s$, and have the same distance to $u, w$, and $z$, but then it contradicts the selection of $s$ because $\lp{v} < \lp{s}$.  Since no vertex in $Y$ satisfies the condition of step~8, $\mathrm{dist}(v, z) = \mathrm{dist}(v, u)$ for all $v\in Y\cap N(s)$.

If $\mathrm{dist}(s', u) = \mathrm{dist}(s', z) = \mathrm{dist}(s', w) = q$, then to have $u<_\sigma w$, one vertex in $Y\cap N(s)$ must be visited before $Z$.  But this would force $z$ to be visited before $w$, because $z$ is at distance $q - 1$ to all vertices in $Y\cap N(s)$.  Now that $\mathrm{dist}(s', w) = q + 1$, if $\mathrm{dist}(s', u) = q$, then at least one vertex $v\in Y$ is adjacent to $s'$; it is in $N(s)$ because $\lp{s} \le \lp{s'}$.  But then $\mathrm{dist}(s', z) \le 1 + \mathrm{dist}(v, z) = 1 + q - 1 = q <\mathrm{dist}(s', w)$.  Therefore, $\mathrm{dist}(s', u) = q+1$ as well.  Each vertex in $Y\cup Z$ has distance at least two to $s'$.  Of vertices at distance two to $s'$, one vertex in $Y\cap N(s)$ must be visited before $Z$, but then we have the same contradiction as in the first case of this paragraph.  Therefore, step~9 is also correct and this concludes the proof of correctness.

We now analyze the running of the algorithm.  Steps~1 and 2 can be easily checked in $O(n+m)$ time.  For step~3, it suffices to calculate the distances between $z, w, u$ and all other vertices; this can be done by visiting the maximal cliques one by one.  Steps 4--7 can be done in $O(n)$ time.   Step~8 can be checked in $O(n)$ time: We have already calculated the distance between $z$ and $v$.  Therefore, the total running time is $O(n + m)$.
\end{proof}

\section{Graph searches on general graphs}
We now describe an algorithm for deciding whether a vertex $z$ of a general graph is an \textsc{mcs} end vertex.  For each subset $X\subseteq V(G)\setminus \{z\}$, we define $f(X)$ to be {true} if there exists an \textsc{mcs} visiting $X$ before others, and {false} otherwise. The question whether $z$ can be an end vertex is then simply the value of $f(V(G)\setminus \{z\})$.  For a set $X$ with $f(X)$ is true and $v\not \in X$, let $g(X, v)$ indicate whether there exists a search ordering that visits $v$ after $X$ and before others.  We have 
\[
  f(X) = \bigvee_{v\in X} \big( f(X\setminus \{v\}) \land g(X\setminus \{v\}, v) \big).
\]
  For \textsc{mcs}, $g(X, v)$ can be calculated in linear time, and thus we have a simple $O(2^n n^{O(1)})$-time algorithm similar to the classic Held--Karp algorithm \cite{held-62-dp}.

  Let us consider then \textsc{bfs}.  We may fix the starting vertex $s$, which can be found by enumerating all the other $n - 1$ vertices.  Let $\ell = \max_{v\in V(G)}\mathrm{dist}(s, v)$, and for $1\le i\le \ell$, let $L_i$ denote the set of vertices at distance $i$ to $s$.  
  Suppose that there is a \textsc{bfs} ordering $\sigma$ started with $s$  and ended with $z$, then $z\in L_\ell$.  Clearly, vertices in $L_{\ell - 1}$ are visited after those in $L_{\ell - 2}$ and before $L_{\ell}$.  Let $u$ be the first visited vertex in $L_{\ell - 1}$ that is adjacent to $z$, and let $X$ be those vertices in $L_{\ell - 1}$ visited before $u$.  Since $z$ is the last vertex, all vertices in $L_\ell\setminus N(X)$ must be adjacent to $u$.  We do not need any constraint on the order of vertices in $L_{\ell - 1}\setminus (X\cup \{u\})$ being visited.
Therefore, the information we need at level $\ell - 1$ are the set $X$ and the vertex $u$.  We can generalize this observation to give a recursive formula for the \textsc{bfs} end vertex problem.

\begin{lemma}
  There is a $2^n\cdot n^{O(1)}$-time algorithm for solving the \textsc{bfs} end vertex problem.
\end{lemma}
\begin{proof}
We use the following algorithm for the fixed starting vertex $s\in V(G)\setminus \{z\}$.
For $X_i\subset L_i$ and $u_i\in L_i\setminus X_i$, where $1\le i< \ell$, we define a function $f(X_i, u_i)$ to be true if and only if there exists a \textsc{bfs} that starts with $s$ and visits $L_i$ in the order of $X_i, u_i, L_i\setminus (X_i\cup \{u_i\})$.  Note that it is true for all sets and vertices when $i = 1$.  For $1\le i< \ell$, we use the following formula

\begin{align}
    \label{eq:dp-bfs}\tag{\textsc{bfs}-subproblem}
  f(X_{i + 1}, u_{i + 1}) = \bigvee_{\substack{X_i\subset L_i\\ u_i\in L_i\setminus X_i}} \big[& f(X_{i}, u_{i}) \land (X_i\cap N(Y_{i+1}) = \emptyset) \land
  \\
  \notag
 & (X_{i +1}\cup \{u_{i+1}\}\setminus N(X_i)\subseteq N(u_i)) \big]
\end{align}
to calculate $f(X_{i + 1}, u_{i + 1})$, where $Y_{i + 1} = L_{i + 1}\setminus X_{i + 1}$.
After that, we return $f(L_\ell\setminus\{z\}, z)$.

We now show the correctness of \eqref{eq:dp-bfs}.
  Suppose that there is a \textsc{bfs} ordering $\sigma$ of $G$ that visits vertices in $L_{i + 1}$ in the order of $X_{i + 1}$, $u_{i + 1}$, and $L_{i + 1}\setminus (X_{i + 1}\cup \{u_{i + 1}\})$.  Let $u_i$ be the first vertex in $L_i\cap N(Y_{i+1})$ visited in $\sigma$, and $X_i$ those vertices in $L_{i}$ visited before $u_i$.  By the selection of $u_i$ and $X_i$, the value of $f(X_{i}, u_{i})$ is true and $X_i\cap N(Y_{i+1})$ is empty.  Since $u_i$ is adjacent to some vertex in $Y_{i+1}$, to enforce the order in $L_{i+1}$, vertex $u_i$ has to be adjacent to $u_{i + 1}$ and all vertices in $X_{i+1}\setminus N(X_i)$.  Therefore, the right-hand side is true for $X_i$ and $u_i$.

  The other direction is similar.  Suppose that $\sigma$ is a \textsc{bfs} ordering of $G$ that visits vertices in $L_{i}$ in the order of $X_{i}$, $u_{i}$ before others, and the three conditions are all true.  After finishing the $L_i$, we proceed as follows. In level $i+1$, we visit $N(X_i)\cap L_{i+1}\subseteq X_{i + 1}$, and then all the other vertices in $X_{i + 1}$ and $u_{i + 1}$ in order, which are adjacent to $u_i$.  This verifies that $f(X_{i + 1}, u_{i + 1})$ is true.
\end{proof}

In the last we consider \textsc{dfs}.  Recall that a \textsc{dfs} sets two timestamps for a vertex $v$, first when it is visited, and second when it is \emph{finished}, i.e., when all its neighbors have been examined and the search backtracks to the vertex that discovered $v$ (or terminates when $v$ is the source vertex).  Note that when a vertex is finished,  all its neighbors have been visited, and all but one of them have been finished.  In particular, when the last vertex is visited, no vertex in its neighborhood has been finished.  At any moment, the set of vertices that have been visited but not finished form a path in the depth-first tree.
Suppose that $z$ is the end vertex of a \textsc{dfs} ordering $\sigma$ of $G$.  If $v$ is the earliest visited neighbor of $z$, then all the vertices after $v$ are descendants of $v$ in the depth-first tree.

The following simple property of \textsc{dfs} is stronger than Proposition~\ref{lem:ldfs-component}.  In a \textsc{dfs} ordering $\sigma$, if the set of vertices after $v$, i.e., $\{u: v<_\sigma u\}$, and $v$ induce a connected subgraph, then their visiting order is irrelevant to vertices visited before $v$.
\begin{proposition}
  \label{lem:dfs-component}
  Let $\sigma$ be a \textsc{dfs} ordering of a graph $G$, and let $X$ be the set of last visited $|X|$ vertices in $\sigma$.  The sub-ordering $\sigma|_X$ is a \textsc{dfs} ordering of $G[X]$.  Moreover, if $G[X]$ is connected, then $\sigma$ remains a \textsc{dfs} ordering of $G$ after replacing $\sigma|_X$ with any \textsc{dfs} ordering of $G[X]$ that starts with $\argmin_{v\in X} \sigma(v)$.
\end{proposition}

\begin{lemma}
  There is a $2^n\cdot n^{O(1)}$-time algorithm for solving the \textsc{dfs} end vertex problem.
\end{lemma}
\begin{proof}
  For a vertex set $X\subseteq V(G)$ and $s, t\in X$, we define a function $f(X, s, t)$, which is true only when (1) $G[X]$ is connected; and (2) there exists a \textsc{dfs} ordering of $G[X]$ that starts with $s$ and ends with $t$.   We may assume without loss of generality that $G$ is connected.  Then whether $z$ is a \textsc{dfs} end vertex is $\bigvee_{v\in V(G)\setminus \{z\}} f(V(G), v, z)$.  It is clear that $f(X, s, t)$ is true if $G[X]$ is connected and $s$ is the only neighbor of $t$ in $X$.  We use the following formula to calculate it otherwise
\begin{equation}
  \label{eq:dp-dfs}\tag{\textsc{dfs}-subproblem}
  f(X, s, t) = \bigvee_{\substack{v\in (N(t)\cap X)\setminus \{s\}\\(N[t]\cap X)\setminus \{s\} \subseteq Y}} f\big( (X\setminus Y)\cup\{v\}, s, v \big) \land f(Y, v, t).
\end{equation}


  One direction is clear: If there exist vertex $v$ and set $Y$ such that both $f\left( (X\setminus Y)\cup\{v\}, s, v \right)$ and $f(Y, v, t)$ are true, then we can visit $X\setminus Y$, followed by $v$, and then other vertices in $Y$.  In the rest we focus on the other direction,---i.e., if $f(X, s, t)$ is true, then there must be one $v$ making the right hand true.

Let $v$ be the first visited vertex in $N(t)\cap X\setminus \{s\}$, and let $Y$ comprise $v$ and all the vertices visited after $v$.  The sub-orderings of $\sigma$ restricted to $(X \setminus Y)\cup \{v\}$ and to $Y$ are \textsc{dfs} orderings for $X - (Y\setminus \{v\})$ and $G[Y]$ respectively.  The subgraph $G[X] - (Y\setminus \{v\})$ is connected because all the vertices are descendants of $s$ in the depth-first tree; $G[Y]$ is connected because $t$ cannot be last visited otherwise. 
  By the selection of $v$ and $Y$, $f\left( (X\setminus Y)\cup\{v\}, s, v \right)$ is true.  By Proposition~\ref{lem:dfs-component}, the ordering of visiting $Y\setminus \{v\}$ after $v$ is a \textsc{dfs} ordering of $G[Y]$.  This concludes the proof.
\end{proof}

\paragraph{Acknowledgment.} Y.C.~would like to thank Jing Huang for bringing the end vertex problems to his attention.

\bibliographystyle{plainurl}

\end{document}